\documentclass[11pt]{article}
\usepackage{fullpage}
\usepackage{amsmath}
\usepackage{amssymb}
\usepackage{amsthm}

\newtheorem{lem}{Lemma}
\newtheorem{thm}[lem]{Theorem}
\newtheorem{cor}[lem]{Corollary}
\newtheorem{pro}[lem]{Proposition}

\newcommand{\iex}{\mathbf{i_{\textrm{ext}}}}
\renewcommand{\i}{\mathbf{i}}
\renewcommand{\v}{\mathbf{v}}
\newcommand{\cross}{+}
\newcommand{\im}{\mathrm{im}}
\newcommand{\E}{\mathbb{E}}
\newcommand{\var}{\mathbf{Var}}
\renewcommand{\P}{\mathbb{P}}
\newcommand{\Otilde}{\widetilde{O}}
\newcommand{\onevec}{\mathbf{1}}
\newcommand{\R}{\mathbb{R}}

\newcommand{\Ztilde}{\widetilde{Z}}
\newcommand{\ztilde}{\tilde{z}}

\newcommand{\wmin}{w_{min}}
\newcommand{\wmx}{w_{{max}}}

\def\norm#1{\left\| #1 \right\|}
\def\abs#1{\left|#1  \right|}

\begin{document}
\title{Graph Sparsification by Effective Resistances\thanks{
This material is based upon work supported by the National Science Foundation
under Grants No. CCF-0707522 and CCF-0634957.
Any opinions, findings, and conclusions or recommendations expressed in this material are those of the author(s) and do not necessarily reflect the views of the National Science Foundation.
}} 
\author{
Daniel A. Spielman\\
Program in Applied Mathematics and \\
Department of Computer Science \\
Yale University
\and
Nikhil Srivastava\\
Department of Computer Science \\
Yale University
}

\maketitle
\begin{abstract}
We present a nearly-linear time algorithm that produces high-quality spectral sparsifiers
of weighted graphs. Given as input a weighted graph $G=(V,E,w)$ and a parameter
$\epsilon>0$, we produce a weighted subgraph $H=(V,\tilde{E},\tilde{w})$ of $G$
such that $|\tilde{E}|=O(n\log n/\epsilon^2)$ and for all vectors $x\in\R^V$
\begin{equation}\label{abseq} (1-\epsilon)\sum_{uv\in E}(x(u)-x(v))^2w_{uv}\le \sum_{uv\in
\tilde{E}}(x(u)-x(v))^2\tilde{w}_{uv} \le
(1+\epsilon)\sum_{uv\in E}(x(u)-x(v))^2w_{uv}.\end{equation}
This improves upon the spectral sparsifiers constructed by Spielman and Teng, which had
$O(n\log^c n)$ edges for some large constant $c$, and upon 
  the cut sparsifiers of Bencz\'ur and Karger, which only
satisfied (\ref{abseq}) for $x\in\{0,1\}^V$.

A key ingredient in our algorithm is a subroutine of independent interest: a
nearly-linear time algorithm that builds a data structure from which we can
query the approximate effective resistance between any two vertices in a graph in $O(\log
n)$ time.

\end{abstract}
\section{Introduction}
The goal of sparsification is to approximate a given graph $G$ by a
sparse graph $H$ on the same set of vertices. If $H$ is close to $G$ in some
appropriate metric, then $H$ can be used as a proxy for $G$ in computations
without introducing too much error. At the same time, since $H$ has very few
edges, computation with and storage of $H$ should be cheaper. 

We study the notion of spectral sparsification introduced
  by Spielman and Teng~\cite{st08}.
Spectral sparsification 
  was
  inspired by the notion of cut sparisification introduced by
  Bencz\'ur and Karger~\cite{bk}
to
accelerate cut algorithms whose running time depends on the
number of edges. They gave a nearly-linear time procedure which takes a
graph $G$ on $n$ vertices with $m$ edges and a parameter $\epsilon>0$, and outputs a weighted subgraph $H$
with $O(n\log n/\epsilon^2)$ edges such that the weight of every cut in
$H$ is within a factor of $(1\pm \epsilon)$ of its weight in $G$. This was used 
to turn Goldberg and Tarjan's $\Otilde(mn)$ max-flow algorithm \cite{goldtar} into an $\Otilde(n^2)$ algorithm
for approximate $st$-mincut, and appeared more recently as the first step of
an $\Otilde(n^{3/2}+m)$-time $O(\log^2n)$ approximation algorithm for sparsest cut \cite{krv}.

The cut-preserving guarantee of \cite{bk} is equivalent to satisfying
(\ref{abseq}) for all $x\in\{0,1\}^n$, 
which are the characteristic vectors of cuts. 
Spielman and Teng \cite{st04,st08} devised stronger sparsifiers which extend (\ref{abseq}) to all $x\in\R^n$, 
but have $O(n\log^c n)$ edges for some large constant $c$. They used these sparsifiers to construct preconditioners for
symmetric diagonally-dominant matrices, which led to the first nearly-linear
time solvers for such systems of equations.

In this work, we construct sparsifiers that achieve the same guarantee as Spielman and Teng's
but with $O(n\log n/\epsilon^2)$ edges, thus improving on 
both \cite{bk} and \cite{st04}. Our sparsifiers are subgraphs of the original graph and can be computed
in $\Otilde(m)$ time by random sampling, where the sampling probabilities are given by
the effective resistances of the edges.
While this is conceptually much simpler than the recursive partitioning approach of
\cite{st04}, we need to solve $O (\log n)$ linear systems to compute the effective resistances quickly, and we do this 
  using Spielman and Teng's linear equation solver.

\subsection{Our Results}

Our main idea is to include each edge of $G$ in the sparsifier $H$ with probability
proportional to its effective resistance. The effective resistance of an edge is known
to be equal to the probability that the edge
appears in a random spanning tree of $G$ (see, e.g., \cite{doylesnell} or
\cite{bollobas}), and was proven in \cite{chandra} to be
proportional to the commute time between the endpoints of the edge.
We show how to approximate the effective resistances of edges in $G$
quickly and prove that sampling according to these approximate values yields a good sparsifier.

To define effective resistance, identify $G=(V,E,w)$ with an electrical network on $n$ nodes in which each edge
$e$ corresponds to a link of conductance $w_e$ (i.e., a resistor of resistance $1/w_e$).
Then the effective resistance $R_e$ across an
edge $e$ is the potential difference induced across it when a unit current
is injected at one end of $e$ and extracted at the other end of $e$. 
Our algorithm can now be stated as follows.
\begin{center}
\begin{tabular}[c]{p{6in}}
\hline
$H=\textbf{Sparsify}(G,q)$\\
Choose a random edge $e$ of $G$ with probability $p_e$ proportional to
$w_eR_e$,
and add $e$ to $H$ with weight $w_e/qp_e$.
Take $q$ samples independently with replacement,
summing weights if an edge is chosen more than once.\\
\hline
\end{tabular}
\end{center}

Recall that the {\em Laplacian} of a weighted graph is given by $L=D-A$ where $A$ is the weighted
adjacency matrix $(a_{ij})=w_{ij}$ and $D$ is the diagonal matrix $(d_{ii})=\sum_{j\ne i}w_{ij}$ of weighted degrees.
Notice that the quadratic form associated with $L$ is just $x^TLx=\sum_{uv\in E}(x(u)-x(v))^2w_{uv}$.
Let $L$ be the Laplacian of $G$ and let $\tilde{L}$ be the Laplacian of $H$. Our main theorem is that
if $q$ is sufficiently large, then the quadratic forms of $L$ and $\tilde{L}$ are close.

\begin{thm}\label{mainthm} Suppose $G$ and $H=\mathbf{Sparsify}(G,q)$ have
Laplacians $L$ and $\tilde{L}$ respectively, 
and $1/\sqrt{n}<\epsilon\le 1$. If $q=9C^2n\log n/\epsilon^2$, where $C$ is the constant in Lemma \ref{lemrud} and if $n$ is sufficiently large, 
then with probability at least $1/2$
\begin{equation}\label{xtlx} \forall x\in\R^n\quad (1-\epsilon)x^TLx\le x^T\tilde{L}x\le
(1+\epsilon)x^TLx.\end{equation}
\end{thm}

Sparsifiers that satisfy this condition preserve many properties of the graph.
The Courant-Fischer Theorem tells us that
\[ \lambda_i=\max_{S:\dim(S)=k}\min_{x\in S} \frac{x^TLx}{x^Tx}.\] 
Thus, if $\lambda_{1}, \dots , \lambda_{n}$
  are the eigenvalues of $L$ and $\tilde{\lambda}_{1}, \dots , \tilde{\lambda}_{n}$ are the eigenvalues
  of $\tilde{L}$, then we have
\[
  (1-\epsilon) \lambda_{i} \leq \tilde{\lambda}_{i} \leq (1+\epsilon) \lambda_{i},
\]
and the eigenspaces spanned by corresponding eigenvalues are related.
As the eigenvalues of the normalized Laplacian are given by
\[
 \lambda_i=\max_{S:\dim(S)=k}\min_{x\in S} \frac{x^T D^{-1/2} L D^{-1/2} x}{x^Tx},
\]
and are the same as the eigenvalues of the walk matrix $D^{-1} L$,
  we obtain the same relationship between the eigenvalues of the walk matrix of the original
  graph and its sparsifier.
Many properties of graphs and random walks are known to be revealed by their spectra (see 
  for example
\cite{bollobas,chung,GodsilRoyle}). 
The existence of sparse subgraphs which retain these properties is interesting its own right;
indeed, expander graphs can be viewed as constant degree sparsifiers for the complete graph.

We remark that the condition (\ref{xtlx}) also implies
\[ \forall x\in\R^n\quad \frac{1}{1+\epsilon}x^TL^{+}x\le x^T\tilde{L}^{+}x\le
\frac{1}{1-\epsilon}x^TL^{+}x,\]
where $L^{+}$ is the pseudoinverse of $L$.
Thus sparsifiers also approximately preserve the effective resistances between
vertices, since for vertices $u$ and $v$, the effective resistance between them is given by the formula
  $(\chi_{u} - \chi_{v})^{T} L^{+} (\chi_{u} - \chi_{v})$, where $\chi_{u}$ is the elementary unit vector with a coordinate
  1 in position $u$.

We prove Theorem \ref{mainthm} in Section 3. At the end of Section 3, we prove that
the spectral guarantee (\ref{xtlx}) of Theorem \ref{mainthm} is not harmed too
much if use approximate effective resistances for sampling instead of exact ones(Corollary \ref{approxok}).

In
Section 4, we show how to compute approximate effective resistances in
nearly-linear time, which is essentially optimal.
 The tools we use to do this are Spielman and Teng's
nearly-linear time solver \cite{st04,st06} and the Johnson-Lindenstrauss Lemma
\cite{jl,ach}. Specifically, we  prove the following theorem, in which
$R_{uv}$ denotes the effective resistance between vertices $u$ and $v$.

\begin{thm}\label{th2}  There is an $\Otilde(m (\log r)/\epsilon^2)$
time algorithm which on input $\epsilon>0$ and $G=(V,E,w)$ with
  $r = w_{max} / w_{min}$  computes a 
$(24\log n/\epsilon^2)\times n$ matrix $\Ztilde $
such that with probability at least $1-1/n$
 \[ (1-\epsilon)R_{uv}\le \|\Ztilde (\chi_u-\chi_v)\|^2\le (1+\epsilon)R_{uv}\]
 for every pair of vertices $u,v\in V$.
\end{thm} 
Since $\Ztilde(\chi_u-\chi_v)$ is simply the difference of the corresponding two columns
  of $\Ztilde$,
we can query the approximate effective resistance between any pair of
vertices $(u,v)$ in time $O(\log n/\epsilon^2)$, and for all the edges in time $O(m\log n/\epsilon^2)$.
By Corollary \ref{approxok}, this yields an $\Otilde(m (\log r)/ \epsilon^{2})$
time for sparsifying graphs, as advertised.

In Section 5, we show that $H$ can be made close to $G$  in some additional
  ways which make it more useful for preconditioning systems of linear equations.

\subsection{Related Work}
Batson, Spielman, and Srivastava \cite{bss} have given a deterministic 
  algorithm that constructs sparsifiers of size $O(n/\epsilon^2)$ in 
  $O(mn^3/\epsilon^2)$ time. 
While this is too slow to be useful in applications, it is 
  optimal in terms of the tradeoff between sparsity and quality of 
  approximation and can be viewed as generalizing expander graphs.
Their construction parallels ours in that it reduces the task of spectral sparsification to approximating the matrix
  $\Pi$ defined in Section 3; however, their method for selecting edges is iterative
  and more delicate than the random sampling described in this paper.

In addition to the graph sparsifiers of \cite{bk,bss,st04},
there is a large body of work on sparse \cite{ahk,am} and low-rank
\cite{fkv,am,rudversh,drineas, drineas03pass} 
approximations for general matrices. The
algorithms in this literature provide guarantees of the form $\|A-\tilde{A}\|_2\le\epsilon$, where $A$ is the original
matrix and $\tilde{A}$ is obtained by entrywise or columnwise sampling of $A$. This is 
analogous to 
  satisfying (\ref{abseq}) only for vectors $x$ in the span of the dominant eigenvectors of $A$; thus,
if we were to use these sparsifiers on graphs, they would only preserve the large cuts. Interestingly, our proof
uses some of the same machinery as the low-rank approximation result of Rudelson and Vershynin \cite{rudversh} --- the
sampling of edges in our algorithm corresponds to picking $q=O(n\log n)$ columns at random from a 
certain rank $(n-1)$ matrix of dimension $m\times m$ (this is the matrix $\Pi$
introduced in Section 3).

The use of effective resistance as a distance in graphs has recently gained attention
  as it is often more useful than the ordinary geodesic distance in a graph.
For example, in small-world graphs, all vertices will be close to one another, but those with a smaller
  effective resistance distance are connected by more short paths.
See, for instance \cite{fouss1,fouss2}, which use effective resistance/commute time
as a distance measure in social network graphs.

\section{Preliminaries}
\subsection{The Incidence Matrix and the Laplacian} \label{secincidence}
Let $G=(V,E,w)$ be a connected weighted undirected graph with $n$ vertices and $m$
edges and edge weights $w_e > 0$. If we orient the edges of $G$
arbitrarily, we can write its Laplacian as $L=B^TWB$, where $B_{m\times n}$ is the
{\em signed edge-vertex incidence matrix}, given by
\[ B(e,v)= \left\{\begin{array}{ll} 1 &\textrm{if $v$ is $e$'s head}\\ -1 &\textrm{if
$v$ is $e$'s tail}\\ 0 & \textrm{otherwise}\end{array}\right.\]
and $W_{m\times m}$ is the diagonal matrix with $W(e,e)=w_e$. 
Denote the
row vectors of $B$ by $\{b_e\}_{e\in E}$ and the span of its columns by
$\mathbb{B}=\im(B)\subseteq \R^m$ (also called the {\em cut space} of $G$ \cite{GodsilRoyle}).
Note that $b_{(u,v)}^T=(\chi_v-\chi_u)$.

It is immediate that $L$ is positive semidefinite since
\[ x^TLx=x^TB^TWBx=\|W^{1/2}Bx\|_2^2\ge 0\quad\textrm{ for every $x\in\R^n$.}\]
We also have $\ker(L)=\ker(W^{1/2}B)=\textrm{span}(\onevec)$, since 
\begin{align*}
x^TLx=0 &\iff \|W^{1/2}Bx\|_2^2=0
\\&\iff \sum_{uv\in E}w_{uv}(x(u)-x(v))^2=0
\\&\iff x(u)-x(v)=0\quad\textrm{for all edges $(u,v)$}
\\&\iff \textrm{ $x$ is constant, since $G$ is connected.}\end{align*}

\subsection{The Pseudoinverse}\label{secpseudo}
Since $L$ is symmetric we can diagonalize it and write
\[ L=\sum_{i=1}^{n-1}\lambda_i u_iu_i^T\]
where $\lambda_1,\ldots,\lambda_{n-1}$ are the nonzero eigenvalues of $L$ and
$u_1,\ldots,u_{n-1}$ are a corresponding set of orthonormal eigenvectors. The {\em
Moore-Penrose Pseudoinverse} of $L$ is then defined as
\[ L^\cross=\sum_{i=1}^{n-1}\frac{1}{\lambda_i} u_iu_i^T.\]
Notice that $\ker(L)=\ker(L^\cross)$ and that
\[ LL^\cross=L^\cross L=\sum_{i=1}^{n-1}u_iu_i^T,\]
which is simply the projection onto the span of the
nonzero eigenvectors of $L$ (which are also the eigenvectors of $L^\cross$). Thus, $LL^\cross=L^\cross L$ is the identity on
$\im(L)=\ker(L)^\perp=\mathrm{span}(\onevec)^\perp$. We will rely on this fact
heavily in the proof of Theorem \ref{mainthm}.

\subsection{Electrical Flows}
Begin by arbitrarily orienting the edges of $G$ as in Section 2.1. 
We will use the same notation as \cite{guat} to describe electrical flows on graphs:
for a vector $\iex(u)$ of currents injected at the vertices, let $\i(e)$ be the currents induced in
the edges (in the direction of orientation) and $\v(u)$ the potentials induced at the vertices. By Kirchoff's
current law, the sum of the currents entering a vertex is equal to the amount
injected at the vertex:
\[ B^T\i=\iex.\]
By Ohm's law, the current flow in an edge is equal to the potential difference
across its ends times its conductance:
\[\i = WB\v.\]
Combining these two facts, we obtain
\[ \iex=B^T(WB\v)=L\v.\]
If $\iex\perp \mathrm{span}(\mathbf{1})=\ker(L)$ --- i.e., if the total amount of current injected is equal to the total amount extracted --- then we can write
\[\v=L^\cross\iex\]
by the definition of $L^\cross$ in Section \ref{secpseudo}. 

Recall that the {\em effective resistance} between two vertices $u$ and $v$ is defined as
the potential difference induced between them when a unit current is injected at
one and extracted at the other.
We will derive an algebraic expression for the effective resistance in terms of
$L^\cross$.  To inject and extract a unit current across the endpoints of an edge $e=(u,v)$,
we set $\iex=b_e^T=(\chi_v-\chi_u)$, which
is clearly orthogonal to $\mathbf{1}$. The potentials induced by $\iex $
at the vertices are given by $\v=L^\cross b_e^T$; to measure the potential difference across
$e=(u,v)$, we simply multiply by $b_e$ on the left:
\[ \v(v)-\v(u)=(\chi_v-\chi_u)^T\v=b_eL^\cross b_e^T.\]
It follows that the effective resistance across
$e$ is given by $b_eL^\cross b_e^T$ and that the matrix $BL^\cross B^T$
has as its diagonal entries $BL^\cross B^T(e,e)=R_e$.

\section{The Main Result}
We will prove Theorem \ref{mainthm}. Consider the matrix $\Pi = W^{1/2}BL^\cross B^TW^{1/2}$. Since we know
$BL^+B^T(e,e)=R_e$, the diagonal entries of $\Pi$ are
$\Pi(e,e)=\sqrt{W(e,e)}R_e\sqrt{W(e,e)}=w_eR_e$. $\Pi$ has some notable properties.
\begin{lem}[Projection Matrix]\label{lempi} (i) $\Pi$ is a projection matrix. (ii)
$\im(\Pi)=\im(W^{1/2}B)=W^{1/2}\mathbb{B}$. 
(iii) The eigenvalues of
$\Pi$ are $1$ with multiplicity $n-1$ and $0$ with multiplicity $m-n+1$. 
(iv) $\Pi(e,e)=\|\Pi(\cdot,e)\|^2$. 

\end{lem}
\begin{proof} To see (i), observe that 
\begin{align*}
 \Pi^2 &= (W^{1/2}BL^\cross B^TW^{1/2})(W^{1/2}BL^\cross B^TW^{1/2})
 \\&= W^{1/2}BL^\cross (B^TWB)L^\cross B^TW^{1/2}
 \\&= W^{1/2}BL^\cross LL^\cross B^TW^{1/2} \quad\textrm{ since $L=B^TWB$}
 \\&=W^{1/2}BL^\cross B^TW^{1/2}\\
 &\qquad\textrm{ since $L^\cross L$ is the identity
 on $\im(L^\cross$)} 
 \\&=\Pi. \end{align*} 

For (ii), we have
\[\im(\Pi)=\im(W^{1/2}BL^\cross B^TW^{1/2})\subseteq \im(W^{1/2}B).\]
To see the other inclusion, assume $y\in\im(W^{1/2}B)$. Then we can choose
$x\perp\ker(W^{1/2}B)=\ker(L)$ such that $W^{1/2}Bx=y$. But now
\begin{align*}
\Pi y &= W^{1/2}BL^\cross B^TW^{1/2}W^{1/2}Bx
\\&= W^{1/2}BL^\cross Lx\quad\textrm{since $B^TWB=L$}
\\&= W^{1/2}B x\quad\textrm{since $L^\cross Lx=x$ for $x\perp\ker(L)$}
\\&= y.\end{align*}
Thus $y\in\im(\Pi)$, as desired.

For (iii), recall from Section \ref{secincidence} that
$\dim(\ker(W^{1/2}B))=1$. Consequently, $\dim(\im(\Pi))=\dim(\im(W^{1/2}B))=n-1$.
But since $\Pi^2=\Pi$, the eigenvalues of $\Pi$ are all $0$ or $1$, and as $\Pi $
projects onto a space of dimension $n-1$, it must have exactly $n-1$ nonzero
eigenvalues.

(iv) follows from $\Pi^2(e,e)=\Pi(\cdot,e)^T\Pi(\cdot,e)$, since $\Pi$
is symmetric.
\end{proof}

To show that $H=(V,\tilde{E},\tilde{w})$ is a good sparsifier for $G$, we need to show that the
quadratic forms $x^TLx$ and
$x^T\tilde{L}x$ are close. We start by reducing the problem of preserving $x^TLx$ to that of
preserving $y^T\Pi y$.  This will be much nicer since the eigenvalues of $\Pi$
are all $0$ or $1$, so that any matrix $\tilde{\Pi}$ which approximates $\Pi$ in
the spectral norm (i.e., makes $\|\tilde{\Pi}-\Pi\|_2$ small) also preserves its
quadratic form.

We may describe the outcome of $H=\mathbf{Sparsify}(G,q)$ by
the following random matrix:
\begin{equation}\label{defS} S(e,e)=\frac{\tilde{w_e}}{w_e}=
\frac{\textrm{(\# of times $e$ is sampled)}}{qp_e}.\end{equation}
$S_{m\times m}$ is a nonnegative diagonal matrix and the random entry $S(e,e)$ specifies the `amount' of edge
$e$ included in $H$ by $\mathbf{Sparsify}$. For example $S(e,e)=1/qp_e$ if $e$
is sampled once, $2/qp_e$ if it is sampled twice, and zero if it is not sampled
at all. 
The weight of $e$ in $H$ is now given by $\tilde{w_e}=S(e,e)w_e$, and we
can write the Laplacian of $H$ as:
\[ \tilde{L}=B^T\tilde{W}B=B^TW^{1/2}SW^{1/2}B\] since
$\tilde{W}=WS=W^{1/2}SW^{1/2}$.  The scaling of weights by $1/qp_e$ in $\mathbf{Sparsify}$ implies that $\E
\tilde{w_e}=w_e$ (since $q$ independent samples are taken, each with probability
$p_e$), and thus $\E S=I$ and $\E \tilde{L}=L$. 

We can now prove the following lemma, which says that if $S$ does not
distort $y^T\Pi y$ too much then $x^TLx$ and $x^T\tilde{L}x$ are close.
\begin{lem}\label{lemS} Suppose $S$ is a nonnegative diagonal matrix such that
\[\|\Pi S \Pi - \Pi \Pi\|_2\le \epsilon.\] 
Then 
\[ \forall x\in\R^n\quad(1-\epsilon)x^TLx\le x^T\tilde{L}x\le (1+\epsilon)
x^TLx,\]
where $L=B^TWB$ and $\tilde{L}=B^TW^{1/2}SW^{1/2}B$.
\end{lem}
\begin{proof}
The assumption is equivalent to
\[\sup_{y\in \R^m, y\neq 0}\frac{|y^T\Pi (S - I)\Pi y|}{y^Ty} \le \epsilon\]
since $\|A\|_2=\sup_{y\neq 0} |y^TAy|/y^Ty$ for symmetric $A$. Restricting our
attention to vectors in $\im(W^{1/2}B)$, we have
\[\sup_{y\in \im(W^{1/2}B), y\neq 0} \frac{|y^T\Pi(S-I)\Pi y|}{y^Ty} \le \epsilon.\]
But by Lemma \ref{lempi}.(ii), $\Pi$ is the identity on $\im(W^{1/2}B)$ so $\Pi
y=y$ for all $y\in\im(W^{1/2}B)$. Also, every such $y$ can be written as $y=W^{1/2}Bx$ for
$x\in\R^n$. Substituting this into the above expression we obtain:
\begin{align*}
\lefteqn{\sup_{y\in \im(W^{1/2}B), y\neq 0} \frac{|y^T\Pi(S-I)\Pi y|}{y^Ty}}\\
 &= \sup_{y\in
\im(W^{1/2}B),y\neq 0} \frac{|y^T(S-I)y|}{y^Ty}
\\ &=\sup_{x\in\R^n, W^{1/2}Bx\neq 0}\frac{|x^TB^TW^{1/2}SW^{1/2}Bx-x^TB^TWBx|}{x^TB^TWBx}
\\&=\sup_{x\in\R^n, W^{1/2}Bx\neq
0}\frac{|x^T\tilde{L}x-x^TLx|}{x^TLx}\le\epsilon. \end{align*} 
Rearranging yields the desired conclusion for all $x\notin\ker(W^{1/2}B)$.
When $x\in\ker(W^{1/2}B)$ then $x^TLx=x^T\tilde{L}x=0$ and the claim holds trivially.\end{proof}

To show that $\|\Pi S\Pi-\Pi\Pi\|_2$ is likely to be small we use the following concentration
result, which is a sort of law of large numbers for symmetric rank 1 matrices. It was first proven
by Rudelson in \cite{rudl}, but the version we state here appears in the more
recent paper \cite{rudversh} by Rudelson and Vershynin.  
\begin{lem}[Rudelson \& Vershynin, \cite{rudversh} Thm. 3.1] 
\label{lemrud} Let $\mathbf{p}$ be a probability distribution over
$\Omega\subseteq\R^d$ such
that $\sup_{y\in\Omega}\|y\|_2\le M$ and $\|\E_{\mathbf{p}} yy^T\|_2\le 1$. Let
$y_1\ldots y_q$ be independent samples drawn from $\mathbf{p}$. Then 
\[\E \left\|\frac{1}{q}\sum_{i=1}^q y_iy_i^T-\E yy^T\right\|_2\le
\min\left(CM\sqrt{\frac{\log q}{q}},1\right)\]
where $C$ is an  absolute constant.
\end{lem}

\noindent We can now finish the proof of Theorem \ref{mainthm}.
\begin{proof}[Proof of Theorem \ref{mainthm}] 
$\mathbf{Sparsify}$ samples edges from $G$ independently with replacement,
with probabilities $p_e$ proportional to $w_eR_e$. Since 
$\sum_e w_eR_e=\textrm{Tr}(\Pi)=n-1$ by Lemma \ref{lempi}.(iii),
the actual probability distribution over $E$ is given by $p_e=\frac{w_eR_e}{n-1}$.
Sampling $q$ edges from $G$ corresponds to sampling $q$ columns from $\Pi$, so we can
write
\begin{align*}
\Pi S\Pi &= \sum_{e}S(e,e)\Pi(\cdot,e)\Pi(\cdot,e)^T
\\&=\sum_e \frac{(\#\textrm{ of times $e$ is
sampled})}{qp_e}\Pi(\cdot,e)\Pi(\cdot,e)^T\quad\textrm{by (\ref{defS})}
\\&=\frac{1}{q}\sum_{e} (\#\textrm{ of times $e$ is sampled})\frac{\Pi(\cdot,e)}{\sqrt{p_e}}\frac{\Pi(\cdot,e)^T}{\sqrt{p_e}}
\\&=\frac{1}{q}\sum_{i=1}^q y_iy_i^T\end{align*}
for vectors $y_1,\ldots, y_q$ drawn independently with replacement from the
distribution
\[ y = \frac{1}{\sqrt{p_e}}\Pi(\cdot,e) \quad\textrm{with probability } p_e.\] 
We can now apply Lemma \ref{lemrud}.  The expectation of $yy^T$ is given by
\[ \mathbb{E}yy^T = \sum_e p_e \frac{1}{p_e} \Pi(\cdot,e) \Pi(\cdot,e)^T =
\Pi\Pi = \Pi,\]
so $\|\mathbb{E}yy^T\|_2=\|\Pi\|_2=1$.
We also have a bound on the norm of $y$:
\[ \frac{1}{\sqrt{p_e}}\|\Pi(\cdot,e)\|_2 
= \frac{1}{\sqrt{p_e}}\sqrt{\Pi(e,e)}\\
= \sqrt{\frac{n-1}{R_ew_e}}\sqrt{R_ew_e}
= \sqrt{n-1}.\]
Taking $q=9C^2n\log n/\epsilon^2$ gives:
\[ \E \left\|\Pi S\Pi -\Pi\Pi\right\|_2 = \E \left\|\frac{1}{q}\sum_{i=1}^q
y_iy_i^T-\E yy^T\right\|_2\le
C\sqrt{\epsilon^2 \frac{\log (9C^2n\log n/\epsilon^2) (n-1)}{9C^2n\log n}} \le
\epsilon/2,\]
for $n$ sufficiently large, as $\epsilon$ is assumed to be at least $1/\sqrt{n}$.

By Markov's inequality, we have 
\[ \|\Pi S\Pi-\Pi\|_2\le\epsilon\]
with probability at least $1/2$. By Lemma \ref{lemS}, this completes the proof
of the theorem.
\end{proof}
We now show that using approximate resistances for sampling does not damage the 
  sparsifier very much.
\begin{cor} \label{approxok} Suppose $Z_e$ are numbers satisfying $Z_e\ge
R_e/\alpha$ and $\sum_e w_eZ_e\le \alpha \sum_e w_eR_e$ for some $\alpha\ge 1$. 
If we sample as in $\mathbf{Sparsify}$ but take each edge with probability
$p_e'=\frac{w_eZ_e}{\sum_e w_eZ_e}$ instead of $p_e=\frac{w_eR_e}{\sum_e
w_eR_e}$, then $H$ satisfies:
\[ (1-\epsilon\alpha)x^T\tilde{L}x\le x^TLx\le (1+\epsilon\alpha)x^T\tilde{L}x
\quad\forall x\in\R^n,\]
with probability at least $1/2$.
\end{cor}
\begin{proof} We note that
\[ p_e'=\frac{w_eS_e}{\sum_e w_eS_e} \ge \frac{w_e(R_e/\alpha)}{\alpha \sum_e
w_eR_e} =
\frac{p_e}{\alpha^2}\] and proceed as in the proof of Theorem \ref{mainthm}. The
norm of the random vector $y$ is now bounded by:
\[ \frac{1}{\sqrt{p_e'}}\|\Pi(e,\cdot)\|_2 
\le  \frac{\alpha}{\sqrt{p_e}}\sqrt{\Pi(e,e)}\\
= \alpha\sqrt{n-1}\] which introduces a factor of $\alpha$ into the final bound
on the expectation, but changes nothing else.\end{proof}

\section{Computing Approximate Resistances Quickly}
It is not clear how to compute all the effective resistances $\{R_e\}$ exactly and
efficiently. In this section, we show that one can compute constant factor
approximations to all the $R_e$ in time $\Otilde(m \log r)$. In fact, we do
something stronger: we build a $O(\log n)\times n$  matrix $\Ztilde $ from which the effective resistance
between any two vertices (including vertices not connected by an edge) can be
computed in $O(\log n)$ time.

\begin{proof}[Proof of Theorem \ref{th2}] If $u$ and $v$ are vertices in $G$, then the effective resistance between
$u$ and $v$ can be written as:
\begin{align*}
R_{uv}&=(\chi_u-\chi_v)^TL^\cross (\chi_u-\chi_v)\\
&=(\chi_u-\chi_v)^TL^\cross L L^\cross (\chi_u-\chi_v)\\
&=((\chi_u-\chi_v)^TL^\cross B^TW^{1/2})(W^{1/2}B L^\cross (\chi_u-\chi_v))\\
&=\|W^{1/2}BL^\cross (\chi_u-\chi_v)^2\|_2^2.\end{align*}
Thus effective resistances are just pairwise distances between vectors in
$\{W^{1/2}BL^\cross \chi_v\}_{v\in V}$. By the Johnson-Lindenstrauss Lemma, these
distances are preserved if we project the vectors onto a subspace spanned by
$O(\log n)$ random vectors.
For concreteness, we use the following version of the Johnson-Lindenstrauss Lemma
  due to Achlioptas \cite{ach}.
\begin{lem} Given fixed vectors $v_1\ldots v_n\in \R^d$ and
$\epsilon>0$, let
$Q_{k\times d}$ be a random $\pm 1/\sqrt{k}$ matrix (i.e., independent Bernoulli entries)
with $k\ge 24\log n/\epsilon^2$. Then with probability at least $1-1/n$
\[(1-\epsilon)\|v_i-v_j\|_2^2\le \|Qv_i-Qv_j\|_2^2\le
(1+\epsilon)\|v_i-v_j\|_2^2\] for all pairs $i,j\le n$.
\end{lem}

Our goal is now to compute the projections
$\{QW^{1/2}BL^\cross\chi_v\}$.
We will exploit the linear system solver of Spielman and Teng \cite{st04,st06},
which we recall satisfies:

\begin{thm}[Spielman-Teng]\label{st}
There is an algorithm $x = \mathtt{STSolve}(L,y,\delta)$ which
takes a Laplacian matrix $L$, 
a column vector $y$, and an error
parameter $\delta > 0$, and returns a column vector $x$ satisfying
\[
\|x - L^{\cross} y\|_{L} \leq \epsilon \|L^{\cross} y\|_{L},
\] 
where $\norm{y}_{L} = \sqrt{y^{T} L y}$.
The algorithm runs in expected time $\Otilde \left(m \log(1/\delta) \right)$,
where $m$ is the number of non-zero entries in $L$.
\end{thm}

Let $Z = QW^{1/2}BL^\cross$.
We will compute an approximation $\Ztilde$ by using \texttt{STSolve} to approximately compute the
  rows of $Z$.
Let the column vectors $z_{i}$ and $\tilde{z_i}$ denote the $i$th rows of $Z$
  and $\tilde{Z}$, respectively (so that $z_i$ is the $i$th column of $Z^T$).
Now we can construct the matrix $\Ztilde$ in the following three steps.
\begin{enumerate}
\item Let $Q$ be a random $\pm 1/\sqrt{k}$ matrix of dimension $k\times n$ where
$k=24\log  n/\epsilon^2$. 
\item Compute $Y=QW^{1/2}B$. Note that this takes $2m\times 24\log
n/\epsilon^2+m=\Otilde(m/\epsilon^2)$
time since $B$ has $2m$ entries and $W^{1/2}$ is diagonal.
\item 
  Let  $y_i$, for $1\le i\le k$, denote the rows of $Y$,
  and compute $\ztilde_{i} = \mathtt{STSolve}(L,y_{i},\delta)$ 
  for each $i$.
\end{enumerate}

We now prove that, for our purposes, it suffices to call \texttt{STSolve}
  with
\[
  \delta = \frac{\epsilon }{3}
            \sqrt{\frac{2 (1-\epsilon) \wmin}{(1+\epsilon) n^{3}\wmx}}.
\]

\begin{lem}\label{lem:error2} 
 Suppose  
\[
(1-\epsilon) R_{uv} \leq 
\norm{Z (\chi_{u} - \chi_{v})}^{2}
\leq 
(1+\epsilon) R_{uv},
\]
for every pair $u,v\in V$.
If for all $i$,
\begin{equation}\label{stapprox} \|z_i-\tilde{z}_i\|_L\le\delta
\|z_{i}\|_L,\end{equation}
where
\begin{equation}\label{eqn:stprecision}
  \delta \leq  \frac{\epsilon }{3}
            \sqrt{\frac{2 (1-\epsilon) \wmin}{(1+\epsilon) n^{3}\wmx}}
\end{equation}
then 
\[
(1-\epsilon)^{2} R_{uv} 
\leq 
\|\Ztilde  (\chi_{u} - \chi_{v})\|^{2}
\leq 
(1 + \epsilon)^{2} R_{uv},
\]
 for every $uv$.
\end{lem}
\begin{proof}
Consider an arbitrary pair of vertices $u$, $v$.
It suffices to show that 
\begin{equation}\label{eqn:eps3}
\abs{\norm{Z(\chi_u-\chi_v)}-\|\tilde{Z}(\chi_u-\chi_v)\|} \le
  \frac{\epsilon}{3}\norm{Z(\chi_u-\chi_v)}\end{equation}
since this will imply
\begin{align*} \abs{\norm{Z(\chi_u-\chi_v)}^2-\|{\tilde{Z}(\chi_u-\chi_v)}\|^2} 
  &= 
  \abs{\norm{Z(\chi_u-\chi_v)}-\|\tilde{Z}(\chi_u-\chi_v)\|}\cdot
  \abs{\norm{Z(\chi_u-\chi_v)}+\|\tilde{Z}(\chi_u-\chi_v)\|}
  \\&\le
  \frac{\epsilon}{3}\cdot\left(2+\frac{\epsilon}{3}\right)\norm{Z(\chi_u-\chi_v)}^2.
  \end{align*}

As $G$ is connected, there is a simple path $P$ connecting $u$ to $v$.
Applying the triangle inequality twice, we obtain
\begin{align*}
 \abs{ \norm{Z (\chi_{u} - \chi_{v})} -  \norm{\Ztilde  (\chi_{u} - \chi_{v})}}
& \leq 
   \norm{(Z - \Ztilde ) (\chi_{u} - \chi_{v})} \\
& \leq 
  \sum_{ab \in P} \norm{(Z - \Ztilde ) (\chi_{a} - \chi_{b})}.
\end{align*}
We will upper bound this later term by considering its square:
\begin{align*}
\left( \sum_{ab \in P} \norm{(Z - \Ztilde ) (\chi_{a} - \chi_{b})} \right)^{2}
& \leq 
n \sum_{ab \in P} \norm{(Z - \Ztilde ) (\chi_{a} - \chi_{b})}^{2}
\qquad \text{by Cauchy-Schwarz}
\\
& \leq 
n \sum_{ab \in E} \norm{(Z - \Ztilde ) (\chi_{a} - \chi_{b})}^{2} 
\\
& =
n \norm{(Z - \Ztilde ) B^{T}}_{F}^{2}
\qquad \text{writing this as a Frobenius norm}
\\ & = n \norm{B(Z - \Ztilde )^{T}}_{F}^{2}
\\
& \leq 
\frac{n}{w_{min}} \norm{W^{1/2} B(Z - \Ztilde )^{T}}_{F}^{2}
\qquad \text{since $\|W^{-1/2}\|_2\le 1/\sqrt{w_{min}}$}
\\
& \leq 
\delta^{2} \frac{n}{w_{min}} \norm{W^{1/2} B Z^{T}}_{F}^{2}
\\&\qquad\qquad \text{since $\|W^{1/2}B(z_i-\tilde{z}_i)\|^2\le \delta^2 \|W^{1/2}Bz_i\|^2$ by
(\ref{stapprox})}
\\
& = 
\delta^{2} \frac{n}{w_{min}} 
  \sum_{ab \in E} w_{ab} \norm{Z (\chi_{a} - \chi_{b})}^{2}
\\
& \leq 
\delta^{2} \frac{n}{w_{min}} 
  \sum_{ab \in E} w_{ab} (1 + \epsilon) R_{ab}
\\
& \leq 
\delta^{2} \frac{n (1+\epsilon )}{w_{min}} 
  (n-1)
\qquad \text{ by Lemma~\ref{lempi}.(iii).}
\end{align*}

On the other hand,
\[
  \norm{Z (\chi_{u} - \chi_{v})}^{2}
\geq 
  (1 - \epsilon) R_{uv}
\geq 
  \frac{2 (1-\epsilon)}{n w_{max}},
\]
by Proposition~\ref{pro:resistances}.
Combining these bounds, we have
\begin{align*}
\frac{
\abs{ \norm{Z (\chi_{u} - \chi_{v})} -  \norm{\Ztilde  (\chi_{u} - \chi_{v})}}
}{
  \norm{Z (\chi_{u} - \chi_{v})}
}
&\le 
\delta \left({\frac{n (1+\epsilon )}{w_{min}} (n-1)}\right)^{1/2}\cdot\left(\frac{n
w_{max}}{2(1-\epsilon)}\right)^{1/2}
\\&\le \frac{\epsilon}{3}\qquad\textrm{by (\ref{eqn:stprecision}),}
\end{align*}
as desired.
\end{proof}

\begin{pro}\label{pro:resistances}
If $G = (V,E,w)$ is a connected graph, then for all $u,v \in V$,
\[
   R_{uv} \geq   \frac{2}{n w_{max}}.
\]
\end{pro}
\begin{proof}
By Rayleigh's monotonicity law (see~\cite{bollobas}), each resistance $R_{uv}$ in $G$
	is at least the corresponding resistance $R_{uv}'$ in $G'=\wmx \times
	K_n$ (the complete graph with all edge weights $\wmx$) since $G'$ is obtained 
	by increasing weights (i.e., conductances) of edges in $G$. 
But by symmetry each resistance $R_{uv}'$ in $G'$ is exactly
	\[\frac{\sum_{uv}R_{uv}'}{\binom{n}{2}} =
	\frac{(n-1)/\wmx}{n(n-1)/2}=\frac{2}{n\wmx}.\]
Thus $R_{uv}\ge \frac{2}{n\wmx}$ for all $u,v\in V$.
\end{proof}

Thus the construction of $\Ztilde $ takes $\Otilde(m\log(1/\delta)/\epsilon^2)=\Otilde(m \log r/\epsilon^2)$ time. We can then find the approximate resistance
$\| \Ztilde (\chi_u-\chi_v)\|^2\approx R_{uv}$ for any $u,v\in V$ in $O(\log
n/\epsilon^2)$ time simply by subtracting
two columns of $\Ztilde $ and computing the norm of their difference.
\end{proof}

Using the above procedure, we can compute arbitrarily good approximations to the effective resistances
$\{R_e\}$ which we need for sampling in nearly-linear time. By Corollary~\ref{approxok}, any constant factor approximation 
yields a sparsifier, so we are done.

\section{An Additional Property}
Corollary \ref{approxok} suggests that $\mathbf{Sparsify}$ is quite robust with
respect to changes in the sampling probabilities $p_e$, and that we may be able
to prove additional guarantees on $H$ by tweaking them. In this section, we
prove one such claim. 

The following property is desirable for using $H$ to solve linear systems
(specifically, for the construction of {\em ultrasparsifiers} \cite{st04,st06}, which we will not 
define here):
\begin{equation}\label{ultraprop} \textrm{For every vertex $v\in V,$}\quad
\sum_{e\ni v} \frac{\tilde{w}_e}{w_e}\le 2\deg(v).\end{equation}
This says, roughly, that not too many of the edges incident to any given vertex
get blown up too much by sampling and rescaling.
We show how to incorporate this property into our sparsifiers.
\begin{lem} \label{mindeg} Suppose we sample $q > 4 n \log n / \beta $ edges of $G$ as in
$\mathbf{Sparsify}$ with probabilities that satisfy
\[ p_{(u,v)} \ge \frac{\beta}{n\min(\deg(u),\deg(v))}\] for some constant
$0<\beta<1$. Then with probability at least $1-1/n$,
\[ \sum_{e\ni v} \frac{\tilde{w}_e}{w_e} \le 2\deg(v)\quad\textrm{for all $v\in
V$.}\]
\end{lem}
\begin{proof} 
For a vertex $v$, define i.i.d. random variables $X_1,\ldots, X_q$ by:
\[
 X_i = \left\{\begin{array}{ll} \frac{1}{p_e} & \textrm{if $e\ni v$
  is the $i$th edge chosen}\\ 0 & \textrm{otherwise}\end{array}\right.
\]
so that $X_{i}$ is set to $1/p_{e}$ with probability $p_{e}$ for each
  edge $e$ attached to $v$.
Let
\[ D_v =  \sum_{e\ni v} \frac{\tilde{w_e}}{w_e} = 
\sum_{e\ni v}\frac{\textrm{(\# of times $e$ is sampled)}}{qp_e} =
\frac{1}{q}\sum_{i=1}^q X_i.\]
We want to show that with high probability, $D_v\le 2\deg(v)$ for {\em all} vertices $v$.
We begin by bounding the expectation and variance of each $X_i$:
\begin{align*}
\E X_i &= \sum_{e\ni v} p_e\frac{1}{p_e}=\deg(v)\\
\\\var (X_i) &= \sum_{e\ni v}p_e\left(\frac{1}{p_e^2}-\frac{1}{p_e}\right)
\\&\le\sum_{e\ni v} \frac{1}{p_e}
\\&\le \sum_{(u,v)\ni v} \frac{n\min(\deg(u),\deg(v))}{\beta}\quad\textrm{by
assumption}
\\&\le \sum_{(u,v)\ni v}\frac{n\deg(v)}{\beta}
\\&= \frac{n\deg(v)^2}{\beta}\end{align*}

Since the $X_i$ are independent, the variance of $D_v$ is just
\[ \var(D_v)=\frac{1}{q^2}\sum_{i=1}^q\var(X_i)\le\frac{n\deg(v)^2}{\beta q}.\]
We now apply Bennett's inequality for sums of i.i.d. variables (see, e.g.,
\cite{lugosi}), which says
\[ \P[|D_v-\E D_v|>\E D_v]\le \exp\left(\frac{-(\E D_v)^2}{\var(D_v)(1+\frac{\E
D_v}{q})}\right)\]
We know that $\E D_v=\E X_i=\deg(v)$. Substituting our estimate for $\var(D_v)$
and setting $q\ge 4n\log n/\beta$ gives:
\begin{align*}
\P[D_v>2\deg(v)] &\le \exp\left(\frac{-\deg(v)^2}{\frac{n\deg(v)^2}{\beta q}(1+\frac{\deg(v)}{q})}\right)
\\&\le \exp\left(\frac{-\beta q}{2n}\right)\quad\textrm{since
$1+\frac{\deg(v)}{q}\le 2$}\\
\\&\le \exp\left(-2\log n\right)=1/n^2.\end{align*}
Taking a union bound over all $v$ gives the desired result.\end{proof}

Sampling with probabilities
\[p'_e=p'_{(u,v)}=\frac{1}{2}\left(\frac{\|Zb_e^T\|^2w_e}{\sum_e
\|Zb_e^T\|^2w_e} + \frac{1}{n\min(\deg(u),\deg(v))}\right)\] satisfies the
requirements of both Corollary \ref{approxok} (with $\alpha=2$) and Lemma \ref{mindeg} (with
$\beta=1/2$) and yields a sparsifier with the desired property.
\begin{thm} There is an $\Otilde(m/\epsilon^2)$ time algorithm which on input
$G=(V,E,w),\epsilon>0$ produces a weighted subgraph $H=(V,\tilde{E},\tilde{w})$
of $G$ with $O(n\log n/\epsilon^2)$ edges which, with probability at least
$1/2$, satisfies both (\ref{xtlx}) and
(\ref{ultraprop}).\end{thm}

\end{document}